\documentclass[letterpaper]{article}
\usepackage[utf8]{inputenc}
\usepackage[T1]{fontenc}

\usepackage{pgf,tikz,pgfplots}
\usepackage[sort&compress,numbers]{natbib}%
\usetikzlibrary{positioning,calc,patterns,decorations,backgrounds}

\usepackage{hyperref}
\usepackage{xcolor,soul}
\usepackage{amsmath,amssymb,latexsym, amsopn, amsfonts}
\usepackage{url}
\usepackage{graphicx}

\usepackage{sty/amsmod,sty/symbols,sty/alg2emod}
\usepackage{comment}

\newcommand{\Section}[1]{\section{#1}}
\newcommand{\Subsection}[1]{\noindent {\bf #1}~~~}

\newcommand{\SUBsubsection}[1]{\noindent {\em #1}~~~}%

\newcommand{\remove}[1]{}

\newenvironment{proofsketch}{\noindent {\bf Proof Sketch.}\ }{\qed\par\vskip 4mm\par}
\begin{document}

\title{Self-stabilizing TDMA Algorithms for Wireless Ad-hoc Networks without External Reference
~\thanks{The work of this author was partially supported by the EC, through
project FP7-STREP-288195, KARYON (Kernel-based ARchitecture for
safetY-critical cONtrol). This work appears as a brief
announcement~\cite{DBLP:conf/sss/PetigST13}.}}

\author{Thomas Petig
   \footnote{Computer Science and Engineering, Chalmers
University of Technology, Sweden.%
}\\\texttt{petig@chalmers.se}  \and
Elad M.\ Schiller\footnotemark[2]\\\texttt{elad@chalmers.se} \and
Philippas Tsigas\footnotemark[2] \\\texttt{tsigas@chalmers.se}}
\date{}

\maketitle              %

\begin{abstract}
Time division multiple access (TDMA) is a method for sharing communication
media. In wireless communications, TDMA algorithms often divide the radio time
into timeslots of uniform size, $\xi$, and then combine them into frames of
uniform size, $\tau$. We consider TDMA algorithms that allocate at least one
timeslot in every frame to every node. Given a maximal node degree, $\delta$,
and no access to external references for collision detection, time or
position, we consider the problem of collision-free
self-stabilizing TDMA algorithms that use constant frame size.

We demonstrate that this problem has no solution when the frame size is $\tau <
\max\{2\delta,\chi_2\}$, where $\chi_2$ is the
chromatic number for distance-$2$ vertex coloring.
As a complement to this lower bound, we focus on proving the existence
of collision-free self-stabilizing TDMA algorithms that use
constant frame size of $\tau$. We consider basic settings (no
hardware support for collision detection and no prior clock synchronization),
and the collision of concurrent transmissions from transmitters that are at
most two hops apart. In the context of self-stabilizing systems that have no
external reference, we are the first to study this problem (to the best of our
knowledge), and use simulations to show convergence even with computation time uncertainties.\\

\end{abstract}

\Section{Introduction}
Autonomous and cooperative systems will ultimately carry out risk-related
tasks, such as piloting driverless cars, and liberate mankind from mundane
labor, such as factory and production work. Note that the implementation of
these cooperative systems implies the use of wireless ad hoc networks and
their critical component -- the {\em medium access control} (MAC) layer. Since
cooperative systems operate in the presence of people, their safety
requirements include the provision of real-time guarantees, such as constant
communication delay. Infrastructure-based wireless networks successfully
provide high bandwidth utilization and constant communication delay. They
divide the radio into {\em timeslots} of uniform size, $\xi$, that are then
combined into {\em frames} of uniform size, $\tau$. Base-stations, access
points or wireless network coordinators can schedule the frame in a way that
enables each node to transmit during its own timeslot, and arbitrate between
nearby nodes that wish to communicate concurrently. We strive to provide the
needed MAC protocol properties, using limited radio and clock settings, i.e.,
no external reference for collision detection, time or position. Note that ad hoc networks often do not consider collision detection mechanisms, and external references are subject to signal loss. For these settings, we demonstrate that there is no solution for the studied problem
when the frame size is $\tau <\max\{2\delta,\chi_2\}$, where  $\delta$
is a bound on the node degree, and $\chi_2$ is the chromatic number for
distance-$2$ vertex coloring. The main result is the existence of
collision-free self-stabilizing TDMA algorithms that use
constant frame size of $\tau >
\max\{4\delta,X_2\}+1$, where $X_2\geq \chi_2$ is a number that depends on
the coloring algorithm in use. To the best of our knowledge, we are the
first to study the problem of self-stabilizing TDMA timeslot allocation
without external reference. The algorithm simulations demonstrate feasibility in a way that is close to the practical realm.

Wireless ad hoc networks have a dynamic nature that is difficult to predict.
This gives rise to many fault-tolerance issues and requires efficient
solutions. These networks are also subject to transient faults due to temporal
malfunctions in hardware, software and other short-lived violations of the
assumed system settings, such as changes to the communication graph topology.
We focus on fault-tolerant systems that recover after the occurrence of
transient faults, which can cause an arbitrary corruption of the system state
(so long as the program's code is still intact). These {\em
self-stabilizing}~\cite{Dolev:2000} design criteria simplify the task of the
application designer when dealing with low-level complications, and provide an
essential level of abstraction. Consequently, the application design can
easily focus on its task -- and knowledge-driven aspects.

ALOHAnet protocols~\cite{abramson1985da} are pioneering MAC algorithms that
let each node select one timeslot per TDMA frame at random. In the Pure Aloha
protocol, nodes may transmit at any point in time, whereas in the Slotted
Aloha version, the transmissions start at the timeslot beginning. The latter
protocol has a shorter period during which packets may collide, because each transmission can collide only with transmissions that occur
within its timeslot, rather than with two consecutive timeslots as in the Pure
Aloha case. Note that the random access approach of ALOHAnet cannot provide
constant communication delay. Distinguished nodes are often used when the
application requires bounded communication delays, e.g., IEEE 802.15.4 and
deterministic self-stabilizing
TDMA~\cite{arumugam2006self,DBLP:journals/comcom/KulkarniA06}. Without such
external references, the TDMA algorithms have to align the timeslots while
allocating them. Existing algorithms~\cite{DBLP:journals/dc/BuschMSY08}
circumvent this challenge by assuming that $\tau/(\Delta +1)\geq 2$, where
$\Delta$ is an upper bound on the number of nodes with whom any node can
communicate with using at most one intermediate node for relaying messages.
This guarantees that every node can transmit during at least one timeslot, $s$, such that no other transmitter that is at most two hops away, also transmits during $s$.
However, the $\tau/(\Delta +1)\geq 2$ assumption implies bandwidth utilization that is up to $\sO(\delta)$ times lower than the proposed algorithm, because $\Delta \in
\sO(\delta^2)$.

As a basic result, we show that $\tau/\delta \geq 2$, and as a complement to
this lower bound, we focus on considering the case of $\tau/\delta \geq 4$. We
present a collision-free self-stabilizing TDMA algorithm that
use constant frame size of $\tau$. We show that it is sufficient
to guarantee that collision freedom for a single timeslot, $s$, and
a {\em single} receiver, rather than {\em all} neighbors. This narrow
opportunity window allows control packet exchange, and timeslot alignment.
After convergence, there are no collisions of any kind, and each frame
includes at most one control packet.

\Subsection{Related work}
Herman and Zhang~\cite{HermanZhang:2006} assume constant bounds on the communication delay and present self-stabilizing clock synchronization algorithms for wireless ad hoc networks. Herman and Tixeuil~\cite{DBLP:conf/algosensors/HermanT04} assume access to synchronized clocks and present the first self-stabilizing TDMA algorithm for wireless ad hoc networks. They use external reference for dividing the radio time into timeslots and assign them according to the neighborhood topology.
The self-stabilization literature often does not answer the causality dilemma of ``which came first, synchronization or communication'' that resembles Aristotle's {\em `which came first, the chicken or the egg?}' dilemma. On one hand, existing clock synchronization algorithms often assume the existence of MAC algorithms that offer bounded communication delay, e.g.~\cite{HermanZhang:2006}, but on the other hand, existing MAC algorithms that provide bounded communication delay, often assume access to synchronized clocks, e.g.~\cite{DBLP:conf/algosensors/HermanT04}.
We propose a bootstrapping solution to the causality dilemma of ``which came first, synchronization or communication'', and discover convergence criteria that depend on $\tau/\delta$.

The {\em converge-to-the-max synchronization} principle assumes that nodes periodically transmit their clock value, $ownClock$. Whenever they receive clock values, $receivedClock>ownClock$, that are greater than their own, they adjust their clocks accordingly, i.e., $ownClock \gets receivedClock$. Herman and Zhang~\cite{HermanZhang:2006} assume constant bounds on the communication delay and demonstrate convergence. Basic radio settings do not include constant bounds on the communication delay. We show that the converge-to-the-max principle works when given bounds on the expected communication delay, rather than constant delay bounds, as in~\cite{HermanZhang:2006}.

The proposal in~\cite{DBLP:conf/iwdc/Herman03} considers shared variable
emulation. Several self-stabilizing algorithms adopt this abstraction, e.g., a
generalized version of the dining philosophers problem for wireless networks
in~\cite{DBLP:journals/taas/DanturiNT09}, topology discovery in anonymous
networks~\cite{Masuzawa:2009}, random distance-$k$ vertex
coloring~\cite{DBLP:conf/icpads/MittonFLST06}, deterministic distance-$2$
vertex coloring~\cite{DBLP:journals/tcs/BlairM12}, two-hop conflict
resolution~\cite{DBLP:conf/sss/PomportesTBV10}, a transformation from central
demon models to distributed scheduler ones~\cite{DBLP:conf/eurongi/TurauW06},
to name a few.
The aforementioned algorithms assume that if a node transmits infinitely many
messages, all of its communication neighbors will receive infinitely many of
them. We do not make such assumptions about {\em (underlying) transmission
fairness}. We assume that packets, from transmitters that are at most two hops
apart, can collide {\em every time}.

The authors of~\cite{DBLP:conf/algosensors/LeonePS09} present a MAC algorithm
that uses convergence from a random starting state (inspired by
self-stabilization).
In~\cite{DBLP:conf/sss/LeonePSZ10,DBLP:conf/vtc/MustafaPSTT12}, the
authors use computer network simulators for evaluating self-$\star$ MAC
algorithms. A self-stabilizing TDMA algorithm, that accesses external time
references, is presented in~\cite{DBLP:journals/corr/abs-1210-3061}.
Simulations are used for evaluating the heuristics of
MS-ALOHA~\cite{DBLP:conf/vtc/ScopignoC09} for dealing with timeslot exhaustion
by adjusting the nodes' individual transmission signal strength. We provide
analytical proofs and consider basic radio settings. The results presented
in~\cite{DBLP:conf/icdcit/JhumkaK07,DBLP:conf/broadnets/DemirbasH06} do not
consider the time it takes the algorithm to converge, as we do. We mention
a number of MAC algorithms that consider onboard hardware support, such as
receiver-side collision
detection~\cite{DBLP:conf/broadnets/DemirbasH06,DBLP:conf/vtc/ScopignoC09,CS09,YB07,DBLP:journals/dc/BuschMSY08}.
We consider merely basic radio technology that is commonly used in wireless ad
hoc networks.
The MAC algorithms in~\cite{YB07,DBLP:conf/algosensors/ViqarW09} assumes the
accessibility  of an external time or geographical references or the node
trajectories, e.g., Global Navigation Satellite System (GNSS). We instead
integrate the TDMA timeslot alignment with clock synchronization.

\Subsection{Our contribution}
Given a maximal node degree, $\delta$, we consider the problem of the existence of collision-free self-stabilizing TDMA algorithms that use constant frame size of $\tau$. In the context of self-stabilizing systems that have no external reference, we are the first to study this problem (to the best of our knowledge). The proposed self-stabilizing and bootstrapping algorithm answers the causality dilemma of synchronization and communication.

For settings that have no assumptions about fairness and external reference existence, 
we establish a basic limit on the bandwidth utilization of TDMA algorithms in wireless ad hoc networks (Section~\ref{s:bsc}). Namely, $\tau <\max\{2\delta,\chi_2\}$, where
$\chi_2$ is the chromatic number for distance-$2$ vertex
coloring. %
We note that the result holds for general graphs with a clearer connection to bandwidth utilization for the cases of tree graphs ($\chi_2 =\delta+1$) and planar
graphs~\cite{DBLP:journals/jct/MolloyS05} ($\chi_2
=5\delta/3 +\sO(1)$).

We prove the existence of collision-free self-stabilizing TDMA algorithms that
use constant frame size of $\tau$ without assuming the
availability of external references (Section~\ref{s:alg}).
The convergence period is within $\sO(\diam\cdot\tau^2\delta + \tau^4\delta^2)$ steps
starting from an arbitrary configuration, where $\diam$ is the network
diameter. We note that in case the system happens to have access to external
time references, i.e., start from a configuration in which clocks are
synchronized, the convergence time is within $\sO(\tau^3)$, and
$\sO(\tau^3\delta)$ steps when $\tau>2\Delta$, and respectively,
$\tau>\max\{4\delta,\Delta+1\}$. 
We also demonstrate convergence via simulations that take uncertainties into account, such as (local) computation time.

\Section{System Settings}
\label{s:sys}
The system consists of a set, $P := \{ p_i \}$, of communicating entities,
which we call {\em nodes}. An upper bound, $\nu >  | P |$, on the number of
nodes in the system is known. Subscript font is used to point out that $X_i$
is $p_i$'s variable (or constant) $X$. Node $p_i$ has a unique identifier,
$id_i$, that is known to $p_i$ but not necessarily by $p_j \in P \setminus \{
p_i \}$.

\Subsection{Communication graphs}
At any instance of time, the ability of any pair of nodes to communicate, is
defined by the set, $\delta_i \subseteq P$, of {\em (direct) neighbors} that
node $p_i \in P$ can communicate with directly. The system can be
represented by an undirected network of directly communicating nodes,
$G:=(P, E)$, named the {\em communication graph}, where $E:= \{ \{p_i, p_j\}
\in P \times P : p_j \in \delta_i \}$. We assume that $G$ is connected.
For $p_i, p_j \in P$, we define the distance, $d(p_i, p_j)$, as the number of
edges in an edge minimum path connecting $p_i$ and $p_j$. We denote by
$\Delta_i:=\{p_j \in P:0<d(p_i, p_j) \le 2\}$ the $2$-neighborhood of
$p_i$, and the upper bounds on the sizes of $\delta_i$ and $\Delta_i$ are denoted
by $\delta \geq \max_{p_i \in P}( | \delta_i | )$, and respectively,
$\Delta \geq \max_{p_i \in P}( | \Delta_i | )$. We assume that $\diam \ge
\max_{p_i, p_j \in P}d(p_i, p_j)$ is an upper bound on the network diameter.

\Subsection{Synchronization}
The nodes have fine-grained clock hardware (with arbitrary clock offset upon
system start). For the sake of presentation simplicity, our work considers
zero clock skews. %
We assume that the {\em clock} value, $C \in [0,c-1]$, and any timestamp in
the system have $c$ states. The pseudo-code uses the $GetClock()$ function
that returns a timestamp of $C$'s current value. Since the clock value can
overflow at its maximum, and wrap to the zero value, arithmetic expressions
that include timestamp values are module $c$, e.g., the function
$AdvanceClock(x):= C \gets (C + x) \bmod c$ adds $x$ time units to clock
value, $C$, modulo its number of states, $c$. We assume that the maximum clock
value is sufficiently large, $c\gg\diam\tau^2$, to guarantee convergence of
the clock synchronization algorithm, before the clock wrap around.
We say that the clocks are {\em synchronized} when $\forall p_i,p_j \in P:C_i=C_j$, where $C_i$ is $p_i$'s clock value.

Periodic pulses invoke the MAC protocol, and divide the radio time into {\em
(broadcasting) timeslots} of $\xi$ time units in a way that provides
sufficient time for the transmission of a single packet. We group $\tau$
timeslots into {\em (broadcasting) frames}. The pseudo-code uses the event
$timeslot(s)$ that is triggered by the event $0=C_i\bmod \xi$  and $s :=
C_i\div\xi\bmod \tau$ is the {\em timeslot number}, where $\div$ is the
integer division.

\Subsection{Operations}
The communication allows a message exchange between the sender and the receiver.
After the sender, $p_i$, fetches message $m \gets MAC\_fetch_i()$ from the
upper layer, and before the receiver, $p_j$, delivers it to the upper layer in
$MAC\_deliver_j(m)$, they exchange $m$ via the operations $transmit_i(m)$, and
respectively, $m \gets receive_j()$. We model the communication channel,
$q_{i,j}$ (queue), from node $p_i$ to node $p_j \in \delta_i$ as the most
recent message that $p_i$ has sent to $p_j$ and that $p_j$ is about to
receive, i.e., $|q_{i,j}| \leq 1$. When $p_i$ transmits message $m$, the
operation $transmit_i(m)$ inserts a copy of $m$ to every $q_{i,j}$, such that
$p_j \in \delta_i$. Once $m$ arrives, $p_j$ executes $receive()$ and returns
the tuple $\langle i, t_i, t_j, m \rangle$, where $t_i = C_i$ and $t_j = C_j$
are the clock values of the associated $transmit_i(m)$, and respectively, $m
\gets receive_j()$ calls. We assume zero propagation delay and efficient
time-stamping mechanisms for $t_i$ and $t_j$.
Moreover, the timeslot duration, $\xi$, allows the transmission and reception
of at least a single packet, see Property~\ref{def:collisionProperty}.

\begin{property}\label{def:collisionProperty}
   Let $p_i\in P$, $p_j\in\delta_i$. At any point in time $t_i$ in which node
   $p_i$ transmits message $m$ for duration of $\xi$, node $p_j$ receives
   $m$ if there is no node $p_k\in(\delta_i\cup\delta_j)\setminus\{p_i\}$ that
   transmits starting from time $t_k$ with duration $\xi$ such that %
   $[t_i,t_i+\xi)$ and $[t_k,t_k+\xi)$ are intersecting.
\end{property}
This means a node can receive a message if no node in the neighborhood of
      the sender and no node in the neighborhood of the receiver is
transmitting concurrently.

\Subsection{Interferences}
Wireless communications are subject to interferences when two or more
neighboring nodes transmit {\em concurrently}, i.e., the packet transmission
periods overlap or intersect. We model communication interferences, such as
unexpected peaks in ambient noise level and concurrent transmissions of
neighboring nodes, by letting the {\em (communication) environment} to
selectively omit messages from the communication channels. We note that we do {\em not} consider any error (collision) indication from the environment.

The environment can use the operation $omission_{i,j}(m)$ for removing message
$m$ from the communication channel, $q_{i,j}$, when $p_i$'s transmission of $m$
to $p_j \in \delta_i$ is concurrent with the one of $p_k \in \Delta_i$.
Immediately after $transmit_i(m)$, the environment selects a subset of $p_i$'s
neighbors, $Omit_m \subseteq \delta_i$, removes $m$ from $q_{i,j} : p_j \in
Omit_m$ and by that it prevents the execution of  $m \gets receive_j()$. Note
that $Omit_m = \delta_i$ implies that no direct neighbor can receive message
$m$.

\Subsection{Self-stabilization}
Every node, $p_i \in P$, executes a program that is a sequence of {\em
(atomic) steps}, $a_i$. The state, $st_i$, of node $p_i\in P$ includes $p_i$'s
variables, including the clocks and the program control variables, and the
communication channels, $q_{i,j} : p_j \in \delta_i$. The {\em{} (system)
configuration} is a tuple $c:=(st_1,\ldots,st_{|P|})$ of node states. Given a
system configuration, $c$, we define the set of {\em applicable steps}, $a=\{
a_i \}$, for which $p_i$'s state, $st_i$, encodes a non-empty communication
channel or an expired timer. An {\em{}execution} is an unbounded alternating
sequence $R:=(c[0],a[0],c[1],a[1],\ldots)$ (Run) of configurations $c[k]$, and
applicable steps $a[k]$ that are taken by the algorithm and the environment.
The task $\sT$ is a set of specifications and $LE$ (legal execution) is the
set of all executions that satisfy $\sT$.
We say that configuration $c$ is {\em safe}, when every execution that starts
from it is in $LE$. An algorithm is called {\em self-stabilizing} if it
reaches a safe configuration within a bounded number of steps.

\Subsection{Task definition}
We consider the task $\sT_{_\mathrm{TDMA}}$, that requires all nodes, $p_i$, to have timeslots, $s_i$, that are uniquely allocated to $p_i$ within $\Delta_i$. We define $LE_{_\mathrm{TDMA}}$ to be the set of legal executions, $R$, for which $\forall p_i \in P : (p_j \in P\Rightarrow C_i =C_j) \land (((s_i \in [0, \tau-1]) \wedge (p_j \in \Delta_i)) \Rightarrow s_i \neq s_j)$ holds in all of $R$'s configurations. We note that for a given finite $\tau$, there are communication graphs for which $\sT_{_\mathrm{TDMA}}$ does not have a solution, e.g., the complete graph, $K_{\tau+1}$, with $\tau+1$ nodes. In Section~\ref{s:bsc}, we show that the task solution can depend on the (arbitrary) starting configuration, rather then just the communication graph.

\Section{Basic Results}
\label{s:bsc}
We establish a basic limitation of the bandwidth utilization for TDMA
algorithms in wireless ad hoc networks. Before generalizing the limitation, we
present an illustrative example (Lemma~\ref{th:existsConfiguration}) of a
starting configuration for which $\tau <\max\{2\delta,\chi_2\}$, where
$\chi_2$ is the chromatic number for distance-$2$ vertex coloring.

\begin{lemma}\label{th:existsConfiguration}
Let $\delta \in \N$ and $\tau < 2\delta$. Suppose that
the communication graph, $G:=(\{ p_0, \ldots p_\delta \}, E)$, has the
topology of a star, where the node $p_\delta$ is the center (root) node and $E
:= \{p_\delta\} \times L$, where $L:=\{ p_0, \ldots p_{\delta-1} \}$ are the
leaf nodes. There is a starting configuration $c[x]$, such that an execution
$R$ starting from $c[x]$ of any algorithm solves the task
$\sT_{_\mathrm{TDMA}}$ does not converge. 
\end{lemma}
\begin{figure}
   \centering
   \includegraphics[width=0.4\textwidth]{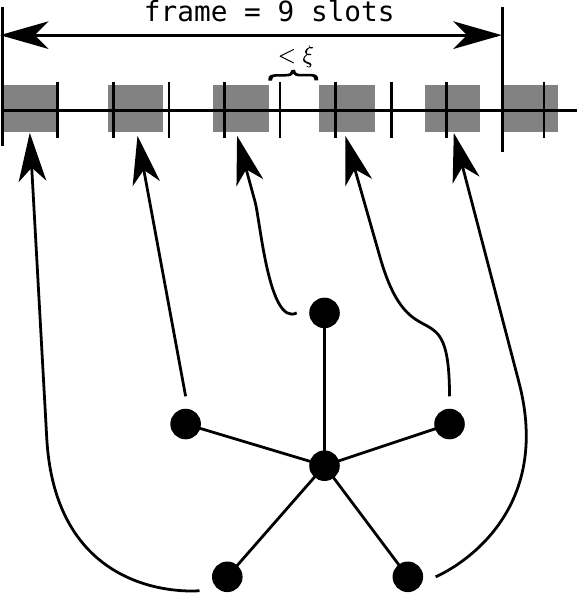}
   \caption{The outer five nodes are covering nine timeslots. The top
   horizontal line and its perpendicular marks depict the radio time division
   according to the central node, $p_{\delta}$. The gray boxes depict the
   radio time covered by the leaf nodes, $p_i \in L$. }\label{fig:star}
\end{figure}
\begin{proof}
We prove this Lemma by showing an example in which we assign timeslots to a
subset of nodes in a way, such that they block each other and, thus, disconnect
the communication graph.

Let $\tau=2\delta-1$. Let $c[x]$
be such that: (1) $C_i$ in $c[x]$ has the properties $(C_i+(2\xi-1)
i)\bmod\xi=0$ and $(C_i+(2\xi-1) i)\div\xi\bmod \tau=s_i$ for all $p_i\in
L\setminus\{p_\delta\}$, (2) $s_\delta=\bot$ and (3) there is no message in
transit. Figure~\ref{fig:star} shows such a graph for $\delta=5$.
This means the next timeslot of node $p_i\in L\setminus\{p_\delta\}$
starts $(2\xi-1)i$ clock steps after $c[x]$. The gap between the time $p_i$'s
timeslot ends and $p_{i+1}$'s timeslot starts is
$(2\xi-1)(i+1)-((2\xi-1)i+\xi)=\xi-1<\xi$ clock steps long and thus smaller
than a time slot. The gap between the next transmission of $p_{\delta-1}$ and
the next next transmission of $p_0$ is $(2\delta-1)\xi+(2\xi-1)0 -
((2\xi-1)(\delta-1)+\xi)=\delta-1<\xi$. This pattern repeats, because only
$p_\delta$ receives these messages transmitted by the leaves and $p_\delta$
does not have a time slot assigned and according to
Property~\ref{def:collisionProperty} any attempt of $p_\delta$ in transmitting
can fail. Thus, no algorithm can establish communication here. 
\end{proof}

The proof of Lemma~\ref{th:existsConfiguration} considers that case of $\tau
<2\delta$ and the star topology graph. We note that the same
scenario can be demonstrated in every graph that includes a node with the
degree $\delta$. Thus, we can establish a general proof for $\tau
<\max\{2\delta,\chi(G^2)\}$ using the arguments in the proof of
Lemma~\ref{th:intervalCoverage}, where $\chi_2$ is the chromatic number when
considering distance-$2$ coloring.

\begin{lemma}\label{th:intervalCoverage}
   Let $\xi\in\R,\tau\in\N$ and $S:=\{[a\xi,(a+1)\xi)$$:a\in$$[0,\tau-1]\}$ be
       a partition of $[0,\xi\tau)$. The intervals
           $C:=\{[b_i,b_i+\xi)$$:b_i\in$$\R\}_i$ intersects maximum $2|C|$
               elements of $S$.
\end{lemma}
\begin{proof}
   Suppose that $[b,b+\xi)\in C$ intersects $I:=[a\xi,(a+1)\xi)\in S$ for some
           $a$. Either $I=[b,b+\xi)$, $b\in I$ or $b+\xi\in I$. Therefore, any
               element $[b_i,b_i+\xi)$ of $C$ intersects maximum $2$ elements
                   of $S$, one that contains $b_i$ and one that contains
                   $b_i+\xi$.
\end{proof}

\Section{Self-stabilizing TDMA Allocation and Alignment Algorithm}
\label{s:alg}
We propose Algorithm~\ref{alg:timeslot} as a self-stabilizing algorithm for
the $\task$ task. The nodes transmit data packets, as well as control packets.
Data packets are sent by $\act$ nodes during their data packet timeslots. The
$\psv$ nodes listen to the $\act$ ones and do not send data packets. Both
$\act$ and $\psv$ nodes use control packets, which include the reception time
and the sender of recently received packets from direct neighbors. Each node
aggregates the frame information it receives. It uses this information for
avoiding collisions, acknowledging packet transmission and resolving hidden
node problems.
A passive node, $p_i$, can become $\act$ by selecting
random timeslots, $s_i$, that are not used by $\act$ nodes. Then $p_i$ sends a control packet
in $s_i$ and waiting for confirmation. Once $p_i$ succeeds, it becomes an
$\act$ node that uses timeslot $s_i$ for transmitting data packets. Node $p_i$
becomes $\psv$ whenever it learns about conflicts with nearby nodes, e.g., due
to a transmission failure.

The hidden node problem refers to cases in which node $p_i$ has two neighbors,
$p_j,p_k\in \delta_i$, that use intersecting timeslots. The algorithm uses
random back off techniques for resolving this problem in a way that assures at
least one successful transmission from all $\act$ and $\psv$ nodes within
$\sO(\tau)$, and respectively, $\sO(1)$ frames in expectation. The $\psv$
nodes count a random number of unused timeslots before transmitting a control
packet.
The $\act$ nodes use their clocks for defining frame numbers. They count down
only during TDMA frames whose numbers are equal to $s_i$, where $s_i \in
[0,\tau-1]$ is $p_i$'s data packet timeslot. These back off processes connect
all direct neighbors and facilitate clock synchronization, timeslot alignment
and timeslot assignment. During legal executions, in which all nodes are
$\act$, there are no collisions and each node transmits one control packet
once every $\tau$ frames.

\LinesNotNumbered\begin{algorithm*}[t!]%
   \caption{\label{alg:timeslot}Self-stabilizing TDMA Allocation, code for node $p_i$}
   $status_i\in\{\act,\psv\}$\tcc*{current node status}
   $s_i\in[0,\tau-1]$\tcc*{current data packet timeslot}
   $wait_i,waitAdd_i\in[0,maxWait]$\tcc*{current back off countdown}
      $FI_i:=\{id_k,type_k,occurrence_k,rxTime_k\}_k \subset \sF\sI$\tcc*{frame information}
      $timeOut$\tcc*{constant, age limit of elements in $FI_i$}
      $BackOff():= $ \textbf{let} $(tmp,r)$ $\gets$ $(waitAdd_i,$ $random([1,$ $3\Delta]));$ \textbf{return} $(r+tmp,$ $3\Delta-r)$\tcc*{reset backoff counter}
      $frame():=(GetClock()\div\xi\tau)\bmod\tau,$\tcc*{the current frame number}
      $Slot(t):=(t\div\xi\bmod\tau), s():=Slot(GetClock())$\tcc*{slot number
of time $t$}
      $Local(set):= \{\langle{}\bullet,\local,\bullet\rangle{}\in
      set\}$\tcc*{dist-1 neighbors in $set$}
      $Used(set):=$ %
                  $\bigcup_{\langle{} \bullet,t_k\rangle{}\in
                  set}[Slot(t_k),Slot(t_k+\xi-1)]$;
      $Unused(set):=%
      [0,\tau-1] \setminus Used(set)$\tcc*{set of (un)used slots}
      $ConflictWithNeighbors(set):=
      (\nexists_{\langle{}id_i\bullet\rangle{}\in set}\lor$
         $s_i\in[Slot(t_i),Slot(t_i+\xi)]\lor$
         $\exists_{\langle{}k,\bullet,rxTime\rangle{}\in set,k\neq id_i}:s_i\in[Slot(rxTime-t_j+t_i),Slot(rxTime-t_j+t_i+\xi)])$\tcc*{check for conflicts}
         $AddToFI(set,o) := FI_i \gets FI_i \cup \{\langle x, y, \remote , z'\rangle: \langle x, y, \bullet, z \rangle \in set, z':=(z +\max\{0,o\})\bmod c, z'\le_{timeOut}C_i\}$\tcc*{$set+FI_i$}
         $IsUnused(s):=$ $s\in Unused(FI_i)$ $\vee$ $(Unused(FI_i)= \emptyset$ $\wedge$ $s\in Unused(Local(FI_i)))$\tcc*{is $s$ an unused slot?}

\nl         \Upon\ $timeslot()$\ \Do {\label{alg:timeslot:event:timeslot}
\nl   \uIf(\tcc*[f]{send data packet}){$s()=s_i\land status_i=\act$}{\label{alg:timeslot:TDMAcond}
 $\transmit(\langle{}status_i,Local(FI_i),MAC\_fetch()\rangle{})$\label{alg:timeslot:TDMAtransmit}
   }
\nl   \ElseIf(\tcc*[f]{check if our frame}){$\lnot(status_i=\act\land frame()\neq s_i)$}{\label{alg:timeslot:CSMAcond}
\nl         \uIf(\tcc*[f]{send control packet}){$IsUnused(s())\land wait_i\le0$}{
\nl            $\transmit(\langle{}status_i,Local(FI_i),0\rangle)$\label{alg:timeslot:CSMAtransmit}\;
$\langle wait_i, waitAdd_i\rangle \gets BackOff()$\tcc*{next control packet
countdown}
\nl            \lIf{$status_i\neq\act$}{$\langle s_i,status_i\rangle\gets \langle s(),\act\rangle$}
        }
\nl        \ElseIf(\tcc*[f]{count down}){$wait_i>0 \land IsUnused((s()-1)\bmod \tau)$}{$wait_i\gets \max\{0,wait_i-1\}$\label{alg:timeslot:countdown}}
      }
      \nl      $FI_i\gets \{\langle{}\bullet,rxTime\rangle{} \in FI_i: rxTime\le_{timeOut}GetClock()\}$\label{alg:timeslot:FIcleanup}\tcc*{remove old entries from $FI_i$}
   }
\BlankLine
\nl   \Upon\ $\langle j,t_j,t_i, \langle status_j,FI_j,m^\prime \rangle\rangle \gets\receive()$\ \Do{\label{alg:timeslot:rx}
\nl     \If(\tcc*[f]{conflicts?}){$ConflictWithNeighbors(FI_j)\land status_i=\act$}{\label{alg:timeslot:detectConflict}
        $\langle{}\langle{}wait_i, waitAdd_i\rangle{},status\rangle{}\gets \langle{}BackOff(),\psv\rangle{}$\label{alg:timeslot:getpassive1}\label{alg:timeslot:getpassive0}\tcc*{get $\psv$}
     }
\nl     \uIf(\tcc*[f]{active node acknowledge}){$status_j=\act$} {\label{alg:timeslot:storePacket}
\nl          \lIf{$m'\neq\bot$}{$FI_i \gets \{\langle{}id_i,\bullet\rangle{}\in FI_i:id_i\neq j\} \cup$ $\{\langle{}j,\msg,$ $\local,$ $t_i\rangle{}\}$}\label{alg:timeslot:storePacket:data}
      }
\nl      \ElseIf(\tcc*[f]{passive node acknowledge}){$t_j= t_i\land Slot(t_j)\not\in Used(FI_i)$}{
   \nl          $FI_i \gets \{\langle{}id_i,\bullet\rangle{}\in FI_i:id_i\neq j\} \cup$ $\{\langle{}j,\welcome,$ $\local,$ $t_i\rangle{}\}$\label{alg:timeslot:FIaddWelcome}\;
      }
      \nl      \If(\tcc*[f]{converge-to-the-max}){$t_i < t_j$}{\label{alg:timeslot:checkClock}
          \nl
          $AdvanceClock(t_j-t_i)$\label{alg:timeslot:AdvanceClock}\tcc*{adjust clock}
          \nl         $FI_i\gets\{\langle{}\bullet,(rxTime+t_j-t_i)\bmod c\rangle{}:\langle{}\bullet,rxTime\rangle{}\in FI_i\}$\label{alg:timeslot:FIadvanceClock}\tcc*{shift timestamps in $FI_i$}
\nl%
$\langle{}\langle{}wait_i,waitAdd_i\rangle ,status_i\rangle{}\gets \langle{}BackOff(),\psv\rangle{}$\label{alg:timeslot:getpassive3}\tcc*{get $\psv$}
      }
\nl      $AddToFI(FI_j,t_i-t_j)$\label{alg:timeslot:AddToFI}\tcc*{Aggregate information on used timeslots}
\nl      \lIf{$m'\neq\bot$}{$MAC\_deliver(m')$}\label{alg:timeslot:rxend}
   }
\end{algorithm*}

\Subsection{Algorithm details}
The node status, $status_i$, is either $\act$ or
$\psv$. When it is $\act$, variable $s_i$ contains $p_i$ timeslot number.

The frame information is the set $FI_i$ $:=$ $\{id_k,$ $type_k,$ $occurrence_k,$
$rxTime_k\}_k$ $\subset$ $\sF\sI\:=$ $\ID\times\{\msg,$ $\welcome\}\times\{\remote,$
$\local\}\times\N$ that contains information about recently received packets,
where $\ID:=\{\bot\}\cup\N$ is the set of possible ids and the null value denoted by $\bot$.
An element of the frame information contains the id of the sender $id_k$. The
type $type_k=\msg$ indicates that the sender was active. For a passive sender
$type_k=\welcome$ indicates that there was no known conflict when this
element was added to the local frame information. If $occurrence_k=\local$,
the corresponding packet was received by $p_i$, otherwise it was copied from a
neighbor. The reception time $rxTime_k$ is the time when this packet was
received, regarding the local clock $C_i$, i.e., it is updated whenever the
local clock is updated. The algorithm considers the frame information to
select an unused timeslot. An entry in the frame information with timestamp $t$
covers the time interval $[t,t+\xi)$.

Nodes transmit control packets according to a random back off strategy for
collision avoidance. The $\psv$ node, $p_i$, chooses a random back off value,
stores it in the variable $wait_i$, and uses $wait_i$ for counting down the
number of timeslots that are available for transmissions. When $wait_i = 0$,
node $p_i$ uses the next unused timeslot according to its frame information.
During back off periods, the algorithm uses the variables $wait_i$ and
$waitAdd_i$ for counting down to zero. The process starts when node $p_i$
assigns $wait_i \gets waitAdd_i+r$, where $r$ is a random choice from
$\left[1,3\Delta \right]$, and updates $waitAdd_i \gets 3\Delta-r$, cf.
$BackOff()$.

The node clock is the basis for the frame and timeslot starting times, cf.
$frame()$, and respectively, $s()$, and also for a given timeslot number, cf.
$Slot(t)$. When working with the frame information, $set$, it is useful to
have restriction by  $\local$ occupancies, cf. $Local(set)$ and to list the
sets of used and unused timeslots, cf.
$Used(set)$, and respectively, $Unused(set)$. We check whether an arriving
frame information, $set$, conflicts with the local frame information that is
stored in $FI_i$, cf. $ConflictWithNeighbors(set)$, before merging them
together, cf. $AddToFI(set,offset)$, after updating the timestamps in $set$,
which follow the sender's clock.

Node $p_i$ can test whether the timeslot number $s$ is available according to
the frame information in $FI_i$ and $p_i$'s clock. Since
Algorithm~\ref{alg:timeslot} complements the studied lower bound
(Section~\ref{s:bsc}), the test in $IsUnused(s)$ checks whether $FI_i$ encodes
a situation in which there are no unused timeslots. In that case,
$IsUnused(s)$ tests whether we can say that $s$ is unused when considering
only transmissions of direct neighbors. The correctness proof considers the
cases in which $\tau>2\Delta$ and $\tau>\max\{4\delta,\Delta+1\}$. For the
former case, Lemma~\ref{th:PeriodFree} shows that there is always an unused
timeslot $s^\prime$ that is not used by any neighbor $p_j\in \Delta_i$,
whereas for the latter case, Lemma~\ref{th:PeriodLocalFree} shows that for any
neighbor $p_j\in \delta_i$, there is a timeslot $s^{\prime\prime}$ for which
there is no node $p_k \in \delta_i\cup\delta_j\cup\{p_j,p_i\}$ that transmits
during $s^{\prime\prime}$.

The code of Algorithm~\ref{alg:timeslot} considers two events: (1) periodic timeslots
(line~\ref{alg:timeslot:event:timeslot}) and (2) reception of a
packet (line~\ref{alg:timeslot:rx}).

\SUBsubsection{{\em (1)} $timeslot()$, line~\ref{alg:timeslot:event:timeslot}:}
Actives nodes transmit their data packets upon their timeslot
(line~\ref{alg:timeslot:TDMAtransmit}). Passive nodes
transmit control packets when the back off counter, $wait_i$, reaches zero
(line~\ref{alg:timeslot:CSMAtransmit}).
Note that $\psv$ nodes count only when the local frame information says that
the previous timeslot was unused (line~\ref{alg:timeslot:countdown}). Active
nodes also send control packets, but rather than counting all unused timeslots,
they count only the unused timeslots that belong to frames with a number that
matches the timeslot number, i.e., $frame()=s_i$
(line~\ref{alg:timeslot:CSMAcond}).

\SUBsubsection{{\em (2)} $\receive()$, line~\ref{alg:timeslot:rx}:}
Active nodes, $p_i$, become $\psv$ when they identify conflicts in $FI_j$
between their data packet timeslots, $s_i$, and data packet timeslots, $s_j$
of other nodes $p_j \in \Delta_i$ (line~\ref{alg:timeslot:detectConflict}).
When the sender is
$\act$, the receiver records the related frame information. Note that the
payload of data packets is not empty in
line~\ref{alg:timeslot:storePacket:data}, c.f., $m^\prime \neq \bot$. Passive nodes, $p_j$, aim to
become $\act$. In order to do that, they need to send a control packet
during a timeslot that all nearby nodes, $p_i$, view as unused, i.e.,
$Slot(t)\not\in Used(FI_i)$, where $t$ is the packet sending time.
Therefore, when the sender is $\psv$, and its data packet timeslots are aligned, i.e.,
$t_i=t_j$, node $p_i$ welcomes $p_i$'s control packet whenever $Slot(t_j)\not\in Used(FI_i)$.
Algorithm~\ref{alg:timeslot} uses a self-stabilizing clock synchronization
algorithm that is based on the converge-to-the-max
principle. When the sender clock value is higher
(line~\ref{alg:timeslot:checkClock}), the receiver adjusts its clock value and
the timestamps in the frame information set, before validating its timeslot, $s_i$, (lines~\ref{alg:timeslot:AdvanceClock}
to~\ref{alg:timeslot:getpassive3}). The receiver can now use the sender's
frame information and payload (lines~\ref{alg:timeslot:AddToFI} to~\ref{alg:timeslot:rxend}).

\Subsection{Correctness}
\label{s:cor}
The proof of
Theorem~\ref{th:selfStab} starts by showing the existence of unused
timeslots by considering the cases in which $\tau>2\Delta$ and
$\tau>\max\{4\delta,\Delta+1\}$ (Lemmas~\ref{th:PeriodFree}, and
respectively,~\ref{th:PeriodLocalFree}). This
facilitates the proof of network connectivity (Lemma~\ref{th:CSMAcomm}), clock
synchronization (Theorem~\ref{th:ConvergeToMax}) and bandwidth allocation
(Theorem~\ref{th:TDMAalloc}).

\begin{theorem}\label{th:selfStab}
   Algorithm~\ref{alg:timeslot} is a self-stabilizing implementation of task
   $\task$ that converges within $\sO(\diam\cdot\tau^2\delta + \tau^4\delta^2)$
   starting from an arbitrary configuration. In case the system happens to have
   access to external time references, i.e., start from a configuration in
   which clocks are synchronized, the convergence time is within
   $\sO(\tau^4)$, and $\sO(\tau^4\delta^2)$ steps when $\tau>2\Delta$, and
   respectively, $\tau>\max\{4\delta,\Delta+1\}$.
\end{theorem}

The proof considers the following definitions.
Given a configuration $c$, we denote by $\sA(c)=\{ p_i\in P: status_i=\act \wedge
wait_i = 0 \}$ the set of active nodes, and by $\sP(c) = P
\setminus \sA(c)$ the set of passive ones.
A \emph{frame} for a node $p_i\in P$ is the time between two successive events of
$timeSlot(s)$ with $s=0$. Note that these frames depend on $C_i$ and, thus,
might not be alinged between nodes.
Communication among neighbors is possible only when there are timeslots that
are free from transmissions by nodes in the local neighborhood.
Lemma~\ref{th:PeriodFree} assumes  that $\tau>2\Delta$ and shows that every
node, $p_i\in P$, has an unused timeslot, $s$, with respect to $p_i$'s clock.
This satisfies the conditions of Property~\ref{def:collisionProperty} with
respect to {\em all} of $p_i$'s neighbors $p_j \in \delta_i$. The proof
considers the definitions of the start and the end of frames and timeslots, as
well as unused timeslots. A configuration, $c^{FrameStart}_i[x_{\ell}] =
c[x]$, in which $GetClock_i()\bmod(\xi\tau)=0$ holds, marks the start of one
of $p_i$'s frames. This frame ends when the next frame starts, i.e., the next
configuration $c^{FrameStart}_i[x_{\ell+1}]$. A timeslot of $p_i$ is,
respectively, bounded by two successive configurations
$c^{TimeSlot}_i[x_{\ell}]$ and
$c^{TimeSlot}_i[x_{\ell+1}]$, such that in those configurations
$GetClock_i()\bmod\xi=0$ holds. The slot number for this timeslot is given as
$GetClock_i()\div\xi\bmod\tau$ at configuration $c^{TimeSlot}_i[x_{\ell}]$. Given
execution $R$, we denote a timeslot starting at $c^{TimeSlot}_i[x_{\ell}]$ as
unused if there is no active node $p_j\in\Delta_i$ exists such that it has in
$R$ an intersecting data packet timeslot. Namely, there is no configuration
$c^{TimeSlot}_j[x'_{\ell}]$ in $R$ with slot number
$GetClock_i()\div\xi\bmod\tau=s_j$ occurs before $c^{TimeSlot}_i[x_{\ell+1}]$ and
a configuration $c^{TimeSlot}_j[x'_{\ell+1}]$ occurs after
$c^{TimeSlot}_i[x_{\ell}]$.
\begin{lemma}
\label{th:PeriodFree}
   Suppose that $\tau>2\Delta$ and $p_i\in P$. Let $R$ be an
   execution of  Algorithm~\ref{alg:timeslot} that includes a complete frame
   start with respect to $p_i$'s clock. Between any two successive frame
   starts, $c^{FrameStart}_i[x_\ell]$ and $c^{FrameStart}_i[x_{\ell+1}]$,
   there is at least one unused timeslot.
\end{lemma}
\begin{proof}
Let us consider all the configurations, $c^\prime$ between
$c^{FrameStart}_i[x_\ell]$, and $c^{FrameStart}_i[x_{\ell+1}]$. Let $S$ be the
partition of $p_i$'s frame in $\tau$ timeslots of length $\xi$ and $C$ be the
maximal set of data packet timeslots of $\act$ nodes $p_j\in \Delta_i$ (and
their respective clocks, $C_j$). We show that $\tau>2\Delta$ implies the
existence of at least one unused timeslot between $c^{FrameStart}_i[x_\ell]$
and $c^{FrameStart}_i[x_{\ell+1}]$, by requiring that $C$'s elements cannot
interest all $\tau$ elements in $S$.
The nodes $p_j$ periodically transmit a data packet once every $\tau$
timeslots (line~\ref{alg:timeslot:TDMAcond}). Note that there are at most
$\Delta$ $\act$ nodes, $p_j$, in all (possibly arbitrary) configurations
$c^\prime$. Namely, every $p_j$ has a single data packet timeslot, $s_j$, but
$s_j$'s timing is arbitrary with respect to $p_i$'s clock. By the proof of
Lemma~\ref{th:intervalCoverage}, $C$ interests maximum $2|C|$ elements of the
set $S$. Since $|S|=\tau$, $|C|\leq\Delta$, and the assumption that $C$'s
elements cannot interest all elements in $S$, we have $\tau>2\Delta$ implies
the existence of at least one unused timeslot between
$c^{FrameStart}_i[x_\ell]$ and $c^{FrameStart}_i[x_{\ell+1}]$.
\end{proof}
Lemma~\ref{th:PeriodFree} is basically the application of
Lemma~\ref{th:intervalCoverage}, were we identify the timeslots with
intervals. 
Lemma~\ref{th:PeriodLocalFree} extends Lemma~\ref{th:PeriodFree} by assuming
that $\tau >\max\{4\delta,\Delta+1\}$ and showing that every node, $p_i\in P$, has an unused
timeslot, $s$, with respect to $p_i$'s clock. This satisfies the conditions of
Property~\ref{def:collisionProperty} with respect to {\em one} of $p_i$'s
neighbors $p_j \in \delta_i$, rather than all $p_i$'s neighbors $p_j \in
\delta_i$, as in the proof of Lemma~\ref{th:PeriodFree}.
\begin{lemma}\label{th:PeriodLocalFree}
   Suppose $\tau >\max\{4\delta,\Delta+1\}$, $p_i\in P$ and $p_j\in\delta_i$. Let $R$ be an
   execution of  Algorithm~\ref{alg:timeslot} that includes a complete frame
   with respect to $p_i$'s clock. With respect to $p_i$'s clock, between any
   two successive frame starts, $c^{FrameStart}_i[x_\ell]$ and
   $c^{FrameStart}_i[x_{\ell+1}]$, there is at least one timeslot that is
   unused by any of the nodes $p_k \in \delta_i\cup\delta_j\cup\{p_j,p_i\}$.
   \end{lemma}
\begin{proof} Let $C$ be the maximum set of data packet timeslots of $\act$
    nodes $p_j\in \delta_i\cap\delta_j\cap\{p_j\}$ (and their respective
    clocks, $C_j$). The proof follows by arguments similar to
    those of Lemma~\ref{th:PeriodFree}. We show that $\tau
    >\max\{4\delta,\Delta+1\}$ implies the existence of at least one timeslot that is unused by
    any of the nodes $p_k \in \delta_i\cup\delta_j\cup\{p_j,p_i\}$ between
    $c^{FrameStart}_i[x_\ell]$ and $c^{FrameStart}_i[x_{\ell+1}]$, by
    requiring that $C$'s elements cannot interest all $\tau$ elements in $S$,
    c.f., proof of Lemma~\ref{th:PeriodFree} for $S$'s definition.
By the proof of Lemma~\ref{th:intervalCoverage}, $C$ interests maximum $2|C|$
elements of the set $S$. Since $|S|=\tau >\max\{4\delta,\Delta+1\}$, $|C|\leq 2\delta$, and the
assumption that $C$'s elements cannot interest all elements in $S$, we have
$\tau >\max\{4\delta,\Delta+1\}$ implies the existence of at least one unused timeslot between
$c^{FrameStart}_i[x_\ell]$ and $c^{FrameStart}_i[x_{\ell+1}]$.
\end{proof}
Lemma~\ref{th:CSMAcomm} shows that the control packet exchange provides
network connectivity.
Recall that Lemmas~\ref{th:PeriodFree}
and~\ref{th:PeriodLocalFree} imply that there is a single timeslot, $s$, that
is unused with respect to the clocks of node $p_i$ and {\em all},
respectively, {\em one} of $p_i$'s neighbors. Lemmas~\ref{th:PeriodFree}
and~\ref{th:PeriodLocalFree} refer to the cases when $\tau > 2\Delta$ and
$\tau > \max\{4\delta,\Delta+1\}$ for which Lemma~\ref{th:CSMAcomm} shows that the
communication delay (during convergence) of the former case is $\delta$ times
shorter than the latter. The proof shows that we can apply the analysis
of~\cite{Hoepman:2011}, because the back off process of a $\psv$ node counts
$r$ unused timeslots, where $r$ is a random choice in $[1,3\Delta]$. The lemma
statement denotes the latency period by $\ell:=(1-e^{-1})^{-1}$.

\begin{dfn}\label{dfn:clockRelation}

   Let $c \in R$ be a configuration and $p_i, p_j \in P: p_j \in \delta_i$ two
   neighbors. We say that the tuple $ e:=\langle j,
   \bullet,\local,time\rangle\in FI_i$ is {\em locally steady} in $c$ when
   $((C_j+\tau\xi-(time\bmod\tau\xi))\div\xi)\div\tau=s_j$. For the case of
   $p_i, p_j, p_k \in P: p_k \in \delta_i \land p_j \in \delta_k\ $ we say
   that  $ e:=\langle j, \bullet,\remote,time\rangle\in FI_i$ is {\em remotely steady}
   in $c$ when $((C_j+\tau\xi-(time\bmod\tau\xi))\div\xi)\div\tau=s_j$.
   A tuple $e\in FI_i$ is {\em stale} if $e$ is neither locally steady, nor
   steady.
\end{dfn}
A frame information set $FI_i$ is locally consistent if all $\local$ entries
can be used to predict a transmission of a node in $\delta_i$. A locally consistent frame
information set $FI_j$ is consistent if all $\remote$ entries can be used to
predict transmissions of nodes in $\Delta_i\setminus\delta_i$.
\begin{lemma}\label{th:CSMAcomm}
   Let $R$ be an execution of Algorithm~\ref{alg:timeslot} that starts from an
   arbitrary configuration $c[x]$. Then there is a suffix $R'$ of $R$ that starts
   from a configuration $c[x+\sO(timeOut)]$ such that a node $p_i\in\sP$
   receives a message from all nodes in $\delta_i$ within finite time.
\end{lemma}
\begin{proof}
   We show that every execution reaches a configuration $c''$ such that in the
   suffix $R'$ of $R$ starting from $c''$ neighbors can exchange packets within
    finite time.
   Lemmas~\ref{th:PeriodFree} and~\ref{th:PeriodLocalFree} show the existence
   of a free time slot for all nodes and their clocks. Then we show that it is
   also marked as free in the node's frame information, since stale entries
   might block it. Therefore, we study the behavior of elements in the frame
   information set $FI_i$ during the different clock update steps.
   Furthermore, we show how elements in the frame information set are
   propagated between neighboring nodes. We see how the distance in time from
   each entry to the current time of a node $C_i$ increases monotonically.
   Thus, a all elements are deleted after $C_i$ advances by $timeOut$ clock
   steps. Additionally, after all stale entries are deleted there is maximal
   one entry in $FI_i$ for each neighbor $p_j\in\Delta_i$. Thus, $p_i$ sees
   the free time slot.

   Note that a node $p_i$ adds only elements $e:=\langle
   id,type,occurrence,rxTime\rangle\in\sF\sI$ to $FI_i$ within the
   lines~\ref{alg:timeslot:storePacket},~\ref{alg:timeslot:FIaddWelcome}
   and~\ref{alg:timeslot:AddToFI}. If we add a direct neighbor in
   line~\ref{alg:timeslot:FIaddWelcome} and~\ref{alg:timeslot:storePacket}
   $rxTime=C_i$ holds.  In latter case, line~\ref{alg:timeslot:AddToFI}, we
   copy entries $e_i:=\langle k,\bullet,\remote,rxTime_i\rangle$ from an
   entries $e_j:=\langle k,\bullet,\remote,rxTime_j\rangle$ direct neighbor
   $p_j\in\delta_i$. But here we update the $rxTime_j$ to $rxTime_i$ according to $p_i$'s
   clock $C_i$ and thus $C_j-rxTime_j\bmod c=C_i-rxTime\bmod c$.

   Let $R$ be a run starting from a configuration $c$, such that $p_i$ has a
   stale entry $e:=\langle j,\bullet,t\rangle\in FI_i$. There are three cases that
   can occur. (1) $p_i$ does not receive a packet during the next
   $timeOut-t+\tau$ clock steps either from $p_j$ directly, or from $p_k$ that
   contains an entry with the id $j$ in the frame information and with a
   smaller clock difference. In this case
   their exists a configuration $c'$ in $R$ such that $e\not\in FI_i(c')$ and
   $0<(C_i(c')-C_i(c)-timeOut+t)\bmod c<\tau$. (2) $p_i$ receives a packet
   from $p_j$ and thus the stale entry $e$ is replaced by a stable entry. (3)
   $p_i$ receives a packet from $p_k$ in configuration $c'$ that contains an entry $e':=\langle
   j,bullet,t'\rangle$, such that $(C_i(c')-t)\bmod c>(C_j(c')-t'$. In any
   case $e'$ replaces $e$. If case $e'$ is stable, the stale entry $e$ is replaced by $e'$. Otherwise $e$
   is replaced by the stale entry $e'$. But note that while replacing $e$ with
   $e'$ the age, i.e., the difference to the nodes clock is maintained. This
   means that the time until it is removed by time out is for $p_i$ and $p_j$
   the same.

   It follows that in a configuration $c''$ that is more then $timeOut$ clock
   steps after $c$, there are only stale entries due to nodes in the
   distance-2 neighborhood that got passive recently, i.e. at most
   $timeOut+\xi$ clock ticks ago. Thus, the free time slot for a node $p_i$,
   that exists by Lemma~\ref{th:PeriodFree} and~\ref{th:PeriodLocalFree}, is
   also free according to $FI_i$. This means $p_i$ can communicate with all
   neighbors in $\delta_i$ within finite time.
\end{proof}
After showing the network connectivity in Lemma~\ref{th:CSMAcomm}, we
continue with bounding the expected communication delay in
Lemma~\ref{th:CSMAcomm2} i.e., how long does it take to successfully send a
control packet to a neighbor.
\begin{lemma}
\label{th:CSMAcomm2}
   Let $R$ be an execution of Algorithm~\ref{alg:timeslot} that starts from an
   configuration $c[x]$ such that a node $p_i\in\sP$ continuously
   receives messages from all nodes in $\delta_i$ within finite time.
   Then every expected $\beta$ frames in $R'$,
   node $p_i\in P$ receives at least one message from all direct $\psv$
   neighbors, $p_j\in\delta_i$, where $\beta$ is $3\Delta\ell$ frames if
   $2\Delta<\tau$, or $3\Delta\ell\delta$ frames if $\tau >\max\{4\delta,\Delta+1\}$. Moreover,
   every expected $\gamma$ frames, $p_i$ receives at least one message from
   all $\act$ neighbor, $p_k\in\delta_i$, where $\gamma$ is
   $3\Delta\tau\ell$ frames if $2\Delta<\tau$, or $3\Delta\tau\ell\delta$ frames
   if $\tau >\max\{4\delta,\Delta+1\}$.
\end{lemma}
\begin{proof}
   Lemmas~\ref{th:PeriodFree} and~\ref{th:PeriodLocalFree} are showing the
   existence of unused time slots. We showed above that the frame information
   set reaches a state in which it is consistent with the actual assignment of
   time slots. It follows that a node has the chance to choose a free time
   slot and exchange packets with all neighbors. We show the expect
   communication delay.
   
   Let $c[x]$ be a configuration.
   Assume node $p_k$ is $\psv$ in $c[x]$. Node $p_k$ counts down a
   random number of unused timeslots regarding $FI_k$.
   Since $FI_k$ does not necessarily contain information about conflicting
   neighbors, i.e., $p_\ell,p_m\in\delta_k\cap\sA(c[x])$ whose data packet timeslots are
   intersecting, $p_k$ could use a slot which is used by some neighbors.
   Suppose that some nodes in $\delta_k\cap\sA(c[x])$ are not in $FI_k$ due to
   conflicts. Then they intersect maximum $2/3$ of the unused timeslots in
   $FI_k$, since two conflicting nodes can only intersect with
   three timeslots, but for each node we add two slots to our frame
   ($\tau>2\Delta$). Therefore, a factor of $3/2$ is added to our expected
   value $\ell$. The same holds if $p_k$ is $\act$.
   In the case of $\tau\in[4\delta+1,2\Delta]$, there are maximum $\delta$
   different timeslots for transmitting packets to all neighbors; one for each
   $p_k\in\delta_i$. The choice of a timeslot is random and, thus, there is an
   additional factor of $\delta$ to hit this timeslot.
\end{proof}

For Theorem~\ref{th:ConvergeToMax} shows that the converge-to-the-max principle works when given bounds on the expectation of the communication delay, rather than constant delay bounds, as in~\cite{HermanZhang:2006}.
\begin{theorem}\label{th:ConvergeToMax}
   Let $R$ be an execution of Algorithm~\ref{alg:timeslot} that starts from an
   arbitrary configuration $c[x]$. Within expected $\Phi:=\gamma\cdot\diam$ frames,
   a configuration $c[x^{synchro}]$ is reached after which all clocks are
   synchronized, where $\gamma$ is the expected communication delay as in
   Lemma~\ref{th:CSMAcomm2}.
\end{theorem}
\begin{proof}

For a moment consider the case of an unbounded clock. Let $M\subset\sP$ the
set of nodes with maximum clock in an arbitrary configuration $c$. Then we
expect every execution to reach a configuration $c'$ within the communication
delay of $\gamma$ frames in which
every neighbor $p_i$ of $M$ has received a message from a node in $M$. Thus,
$p_i$ converged the clock to the maximum value is part of the set of nodes
with maximum clock in $c'$. Note that a node in $M$ is not leaving $M$, since
we neither assume clock skew, nor can it receive a larger clock value.
It follows by induction that the set of nodes with maximum clock value is
increasing monotonically and within expected $\gamma\cdot\diam$ frames $M$
covers $P$ and thus a configuration $c[x^{\textup{synchro}}]$ is reached.

We proceed with investigating bounded clocks. As in~\cite{HermanZhang:2006}
we consider three cases.
In the first case, all clock values, $C_i$, are
smaller than the maximal clock value, $c-1$, by at least the algorithm
convergence time, i.e., $\forall_{p_i}:C_i\in[0,c-1-\Phi]$. The proof
of this case follows that arguments above for the unbounded case.

The second case,
$\forall_{p_i}:C_i\in[0,\Phi-1]\lor C_i[c-\Phi,c-1]$, is that
all clocks are near wrapping around. Clocks can change from the lower interval
to the higher one, but after expected $\Phi$ clock steps, all clocks have
wrapped around, and reached the lower interval $[0,\Phi-1]$, or have
left the lower interval by counting up normally, but not by calling
$AdvanceClock()$. Thus, the proof of the second case is followed by the
arguments of the first case.

The third case supposes that in configuration $c[x]$ there is at least on node whose
clock value is in the range $[\Phi,c-1-\Phi]$, and at least
one with a clock in $[c-1-\Phi,c-1]$. Note that those with a clock in
$[c-1-\Phi,c-1]$ wrap around in maximal $\Phi$ clock steps. Let $c[x']$ be the
configuration after $c[x]$, such that all nodes with a clock value in
$[c-1-\Phi,c-1]$ have wrapped around.

There are three cases. First, all nodes are able to receive a message from one
of the nodes with a clock value in $[c-1-\Phi,c-1]$, before they wrap around
in configuration.Then we have the second case and we are done. If in
configuration $c[x']$ not all nodes receive such a large clock value in
$[c-1-\Phi,c-1]$, then either there is no node in $c[x']$ that has a clock
value in $[c-1-\Phi,c-1]$ and we have the first case, or there is at least one
node with a clock value in $[c-1-\Phi,c-1]$. Let $p_i$ one of the nodes that
have the largest clock value in $c[x']$ and let $c[x'']$ be the first
configuration of $R$ after $c[x]$ such that $p_i$'s clock is in
$[c-1-\Phi,c-1]$. Then, expected $\Phi\cdot\diam$ clock steps after $c[x'']$,
a configuration $c[x''']$ is reached in which all nodes in $\sP$ have either received
a clock that is equal or larger to $p_i$'s clock value and adopted it. Thus,
in $c[x''']$ all nodes have either a clock value in $[c-1-\Phi,c-1]$, or in
$[0,\Phi-1]$, because they have wrapped around. So, this is reduces to the
second case.
\end{proof}
Once the clocks are synchronized, and the TDMA timeslots are aligned, Algorithm~\ref{alg:timeslot} allocates
the bandwidth using distance-$2$ coloring. This happens within $\sO(\gamma^2)$
frames, see Theorem~\ref{th:TDMAalloc}.
\begin{theorem}\label{th:TDMAalloc}
   Let $R$ be an execution that starts from an arbitrary configuration
   $c[x^{synchro}]$ in which all clocks are synchronized. Within expected
   $\gamma^2$ frames from $c[x^{synchro}]$, the system reaches a
   configuration, $c[x^{alloc}]$, in which each node $p_i\in P$ has a
   timeslot that is unique in $\Delta_i$.
\end{theorem}
\begin{proofsketch}
   We have showed that active nodes get feedback within an expected time. Active nodes with positive
   feedback stay active, but conflicting active nodes get a negative feedback
   and change their status to passive.
   Passive nodes are transmitting from time control packets. They are
   successful and stay active on this timeslot with probability $1-1/e$.
   Otherwise, they get a negative feedback within expected
   $\gamma$ frames. Hence, the number of active nodes without
   conflicts is monotonically increasing until every node is active.

   The convergence time of this timeslot assignment is dominated by the time
   for a successful transmission and the time for a negative feedback in case
   of a unsuccessful transmission. This leads to an expected convergence time
   of $\gamma^2$.
\end{proofsketch}
The proof of Theorem~\ref{th:selfStab} is concluded by showing that
configuration, $c[x^{alloc}]$ (Theorem~\ref{th:TDMAalloc}), is a safe
configuration with respect to $\LE$, see Lemma~\ref{th:SafeConfig}.
\begin{lemma}
\label{th:SafeConfig}
    Configuration $c[x]$ is a safe configuration with respect to
    $LE_{_\mathrm{TDMA}}$, when (1) $\forall_{p_i,p_j\in P}:C_i=C_j$, (2)
    $\forall_{p_i\in P}: status_i=\act$, (3) $\forall_{p_i\in
    P}\forall_{p_j\in\Delta_i}: s_i\neq s_j$, (4) $\forall_{ p_i\in
    P}:\forall_{p_j\in\Delta_i\cup\{p_i\}}\exists!_{\langle
    id,\msg,\bullet\rangle\in FI}:id=id_j$.
\end{lemma}
\begin{proof}
   First we check that $c[x]$ is legal regarding $LE_{_\mathrm{TDMA}}$.The
   conditions (1) and (3) are coinciding with the conditions of a legal
   execution $LE_{_\mathrm{TDMA}}$. Condition (2) is necessary in combination
   to (3) to ensure that the TDMA slot stored in $s_i$ is valid. Condition (4)
   is a restriction to general configurations in $LE_{_\mathrm{TDMA}}$ that
   ensures that a node knows about its neighborhood.

   We conclude by proofing that the conditions of this Lemma hold for all
   configurations following $c[x]$ in an execution $R$ of
   Algorithm~\ref{alg:timeslot}.
   Since (1) holds in $c[x]$, this means all clocks are synchronized, a node
   can never receive a clock value that is larger than its own, so the clock
   update step in line~\ref{alg:timeslot:AdvanceClock} is never executed. Thus
   for all following configurations to $c[x]$ in $R$ condition (1) holds.
   A node changes its $status_i$ to $\psv$ when it updates its clock
   (line~\ref{alg:timeslot:AdvanceClock}), or when it detects a conflict with
   the slot assignment (line~\ref{alg:timeslot:detectConflict}). From (1)
   follows that there is no clock update. From (3) follows that in $c[x]$
   there is no conflict with the slot assignment and from (4) that everyone is
   aware that there is no conflict. Furthermore, (4) implies that every node
   selects only slots for control packets that are free within the distance-2
   neighborhood. Thus, in $R$ there is never a message transmitted from that
   the receiver can detect a conflict and no conflict is introduced by sending
   a control packet on a data packet slot of a distance-2 neighbor. This
   proofs that $c[x]$ is a safe configuration for $LE_{_\mathrm{TDMA}}$.
\end{proof}
\begin{cor}\label{th:Safe}
   The configuration $c[x^{alloc}]$ is a safe configuration for the task
   $\sT_{_\mathrm{TDMA}}$.
\end{cor}
\begin{proof}
   We check that $c[x^{alloc}]$ fulfills the conditions of
   Lemma~\ref{th:SafeConfig}. Condition (1) follows from
   Theorem~\ref{th:ConvergeToMax} and conditions (2), (3) and (4) are
   following from Theorem~\ref{th:TDMAalloc}.
\end{proof}
\begin{proof}[Proof of Theorem~\ref{th:selfStab}]
   Lemma~\ref{th:PeriodFree},~\ref{th:PeriodLocalFree} and~\ref{th:CSMAcomm} show that
   communication is possible after a constant number of clock steps.
   Lemma~\ref{th:CSMAcomm2} bounds the expected communication delay to
   $\gamma$. Theorem~\ref{th:ConvergeToMax} shows that after
   $\sO(\gamma\cdot\diam)$ frames a configuration $c[x^{synchro}]$ is
   reached where the clocks are synchronized. Theorem~\ref{th:TDMAalloc} shows
   that in $\sO(\gamma^2)$ frames after $c[x^{synchro}]$, a configuration
   $c[x^{alloc}]$ is reached in which all nodes have an allocated timeslot.
   Corollary~\ref{th:Safe} shows that Algorithm~\ref{alg:timeslot} solves
   $\sT_{_\mathrm{TDMA}}$ by reaching $c[x^{alloc}]$.
\end{proof}
\Section{Experimental results}
\begin{figure}
   \centering
\begin{tikzpicture}[only marks]
    \begin{axis}[xlabel={number of nodes},ylabel={frames},grid=major,%
        width=8cm, height=4cm,
        ymin=0,ymax=150]

        \addplot table[x index=0,y index=2,col sep=comma] {converge.dat};
        \addlegendentry{grid graph}%

        \addplot table[x index=0,y index=1,col sep=comma] {converge.dat};
        \addlegendentry{random graph}

    \end{axis}
\end{tikzpicture}
\caption{The converges time in frames for different graphs. In the grid graph, nodes are
placed on a lattice and connected to their four neighbors. The convergence times are the average over 16
runs that start each with random clock offsets. The random node graph is a unified disk graph
with random node placement with maximal 16 neighbors pair node.}\label{fig:conv}
\end{figure}
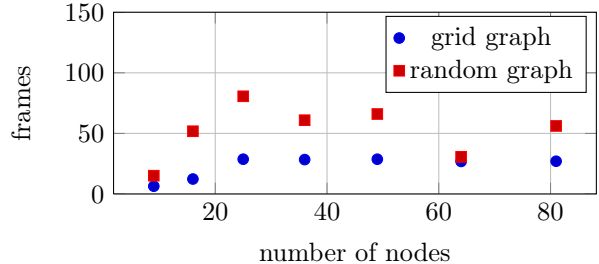
We demonstrate the implementation feasibility. We study the behavior of the
proposed algorithm in a simulation model that takes into account timing
uncertainties. Thus, we demonstrate feasibility in a way that is close to the
practical realm.

The system settings (Section~\ref{s:sys}) that we use for the correctness
proof (Section~\ref{s:cor}) assumes that any (local) computation can be done
in zero time. In contrast to this, the simulations use the TinyOS embedded
operating systems~\cite{levis2005tinyos} and the Cooja simulation
framework~\cite{osterlind2006cross} for emulating wireless sensor nodes
together with their processors. This way Cooja simulates the code execution on
the nodes, by taking into account the computation time of each step. We
implemented the proposed algorithm for sensor nodes that use IEEE 802.15.4
compatible radio transceivers. The wireless network simulation is according to
the system settings (Section~\ref{s:sys}) is based on a grid graph
with $4 \geq \delta$ as an upper bound on the node degree and a random
graph with $16\geq \delta$ as an upper bound on the node degree. The
implementation uses clock steps of $1$ millisecond. We use a time slot size of
$\xi=20$ clock steps, where almost all of this period is spent on
transmission, packet loading and offloading. The frame size is
$\tau=16\ge4\delta$ time slots for the grid graph and $\tau=64\ge4\delta$ for the
random graphs. For these settings, all experiments showed
convergence, see Figure~\ref{fig:conv}.
\Section{Conclusions}
This work considers fault-tolerant systems that have basic radio and clock
settings without access to external references for collision detection, time or
position, and yet require constant communication delay. We study
collision-free TDMA algorithms that have uniform frame size and uniform
timeslots and require convergence to a data packet schedule that does not
change. By taking into account (local) computation time uncertainties, we observe that the algorithm is close to the practical realm. Our analysis considers the timeslot allocation aspects of the studied
problem, together with transmission timing aspects. Interestingly, we show
that the existence of the problem's solution depends on convergence criteria
that include the ratio, $\tau/\delta$, between the frame size and the node
degree. We establish that $\tau/\delta \geq 2$ as a general convergence
criterion, and prove the existence of collision-free TDMA algorithms for which
$\tau/\delta \geq 4$. Unfortunately, our result implies that, for our systems
settings, there is no distributed mechanism for asserting the convergence
criteria within a constant time. For distributed systems that do {\em not}
require constant communication delay, we propose to explore such criteria
assertion mechanisms as future work.

~\\
\Subsection{Acknowledgments}
This work would not have been possible without the contribution of
Marina Papatriantafilou, Olaf Landsiedel and Mohamed H.\ Mustafa in many helpful discussions, ideas, problem definition and
analysis.
\bibliographystyle{abbrv}
\bibliography{bib/localbib}
\end{document}